\documentclass{article}

\usepackage[utf8]{inputenc}
\usepackage{times}
\usepackage[UKenglish]{babel}
\usepackage[T1]{fontenc}
\usepackage{graphicx}
\usepackage{stmaryrd,verbatim}
\usepackage[normalem]{ulem}
\usepackage{tikz-cd} 
\usepackage{csquotes}
\usepackage{framed}
\usepackage[a4paper,top=1.8cm,bottom=1.8cm,left=2cm,right=2cm,marginparwidth=1.75cm]{geometry}
\usepackage{url}
\usepackage{amsmath,amssymb,amsfonts,amsthm,upgreek}
\usepackage{mathrsfs, mathtools}
\usepackage{graphicx}
\usepackage[colorinlistoftodos,prependcaption, textsize=small, textwidth=2.5cm]{todonotes}
\usepackage[all]{xy}
\usepackage{cancel}
\usepackage{amssymb}
\usepackage{appendix}
\usepackage{enumerate}
\usepackage{multicol}
\usepackage{subfigure}
\usepackage{authblk}
\usepackage{xcolor}\usepackage{times}
\usepackage[UKenglish]{babel}
\usepackage[T1]{fontenc}
\usepackage{graphicx}
\usepackage{stmaryrd,verbatim}
\usepackage[normalem]{ulem}
\usepackage{tikz-cd} 
\usepackage{csquotes}
\usepackage{framed}
\usepackage[a4paper,top=1.8cm,bottom=1.8cm,left=2cm,right=2cm,marginparwidth=1.75cm]{geometry}
\usepackage{url}
\usepackage{amsmath,amssymb,amsfonts,amsthm,upgreek}
\usepackage{mathrsfs, mathtools}
\usepackage{graphicx}
\usepackage[colorinlistoftodos,prependcaption, textsize=small, textwidth=2.5cm]{todonotes}
\usepackage[all]{xy}
\usepackage{cancel}
\usepackage{amssymb}
\usepackage{appendix}
\usepackage{enumerate}
\usepackage{multicol}
\usepackage{subfigure}
\usepackage{authblk}
\usepackage{xcolor}

\usepackage[style=numeric, sorting=none, backend=biber, backref=true]{biblatex}

\renewbibmacro*{doi+eprint+url}{
  \printfield{doi}
  \newunit\newblock
  \iftoggle{bbx:eprint}{
    \usebibmacro{eprint}
  }{}
  \newunit\newblock
  \iffieldundef{doi}{
    \usebibmacro{url+urldate}
  }{}
}

\usepackage[colorlinks=true, allcolors=blue
]{hyperref}

\newcommand{\fg}{{\mathfrak g}}

\newcommand{\fk}{{\mathfrak k}}

\newcommand{\fm}{{\mathfrak m}}

\newcommand{\fZ}{{\mathfrak Z}}

\newcommand{\su}{\mathfrak{su}}

\newcommand{\SU}{{\rm SU}}

\newcommand{\Ad}{\mathrm{Ad}}

\renewcommand{\det}{\mathop\mathrm{det}\nolimits}

\renewcommand{\epsilon}{\varepsilon}

\newcommand{\Lie}{\mathrm{Lie}}

\newcommand{\Span}{\mathrm{span}}

\newcommand{\tr}{\mathop{\mathrm{tr}}\nolimits}

\newcommand{\qandq}{\quad\text{and}\quad}

\newcommand{\qforq}{\quad\text{for}\quad}

\def\<{\mathopen{}\left<}
\def\>{\right>\mathclose{}}
\def\({\mathopen{}\left(}
\def\){\right)\mathclose{}}

\usepackage{multicol, color}

\definecolor{gold}{rgb}{0.85,.66,0}
\definecolor{cherry}{rgb}{0.9,.1,.2}
\definecolor{burgundy}{rgb}{0.8,.2,.2}
\definecolor{orangered}{rgb}{0.85,.3,0}
\definecolor{orange}{rgb}{0.85,.4,0}
\definecolor{olive}{rgb}{.45,.4,0}
\definecolor{lime}{rgb}{.6,.9,0}
\definecolor{green}{rgb}{.2,.7,0}
\definecolor{grey}{rgb}{.4,.4,.2}
\definecolor{brown}{rgb}{.4,.3,.1}

\newtheorem{theorem}{Theorem}
\newtheorem{proposition}{Proposition}
\newtheorem{corollary}[proposition]{Corollary}
\newtheorem{lemma}[proposition]{Lemma}
\newtheorem{remark}[proposition]{Remark}

\theoremstyle{definition}
\newtheorem{definition}[proposition]{Definition}

\begin{document}
\title{Cartan-Khaneja-Glaser decomposition of $\SU(2^n)$ via involutive automorphisms}
\author[1,2]{John A. Mora Rodríguez\thanks{j196970@dac.unicamp.br}}
\author[1]{Arthur C. R. Dutra\thanks{a201381@dac.unicamp.br}}
\author[1]{Henrique N. Sá Earp\thanks{hqsaearp@unicamp.br}}
\author[1]{Marcelo {Terra Cunha}\thanks{tcunha@unicamp.br}}

\affil[1]{Instituto de Matemática, Estatística e Computação Científica\par
Universidade Estadual de Campinas (Unicamp)\par
Campinas-SP, Brasil}
\affil[2]{QuIIN - Quantum Industrial Innovation, \par 
EMBRAPII CIMATEC Competence Center in Quantum Technologies,
SENAI CIMATEC\par
Salvador-BA, Brazil}

\date{\today}

\maketitle

\begin{abstract}
We present a novel algorithm for performing the Cartan-Khaneja-Glaser decomposition of unitary matrices in $\SU(2^n)$, a critical task for efficient quantum circuit design. Building upon the approach introduced by Sá Earp and Pachos (2005), we overcome key limitations of their method, such as reliance on ill-defined matrix logarithms and the convergence issues of truncated Baker–Campbell–Hausdorff (BCH) series. Our reformulation leverages the algebraic structure of involutive automorphisms and symmetric Lie algebra decompositions to yield a stable and recursive factorization process. We provide a full Python implementation of the algorithm, available in an open-source repository, and validate its performance on matrices in $\SU(8)$ and $\SU(16)$ using random unitary benchmarks. The algorithm produces decompositions that are directly suited to practical quantum hardware, with factors that can be implemented near-optimally using standard gate sets.
\end{abstract}

\tableofcontents

\newpage
\section{Introduction}

At the heart of quantum algorithms lies the need to implement unitary transformations in the group $\SU(2^n)$, where $n$ is the number of qubits involved, which will evolve the system from an initial state to a desired state\cite{nilsen_chuang}. 
Very celebrated results show that quantum computing can be implemented using only a few quantum gates, acting on one and two qubits only \cites{knill1995approximation,DiVincenzo_1995}.
This is usually referred to in the quantum information community as \emph{universality}: a small set of unitaries which can be used to approximate any given unitary is an \emph{universal set}.
Well-known examples of universal sets are the Control-NOT and a well-chosen small subset of $\SU(2)$, cf.  \cite{EkertNotes}.

While universality guarantees the existence of such approximations, it does not provide efficient or practical means for decomposing arbitrary unitaries into elementary gates. 
As it is typical from computer science, universality is an asymptotic result.
Realistic projects of quantum computers can take advantage of other choices of building blocks
The problem of finding optimal or structured decompositions becomes central to quantum circuit design \cite{PhysRevLett.92.177902,PhysRevLett.93.130502}. 
In this context, Khaneja and Glaser  introduced a promising approach, using the Cartan decomposition in the $\SU(2^n)$ group to express a unitary transformation as a product of single-qubit evolution and multi-qubit Abelian elements \cite{khaneja2000}. 

In 2005, the third-named author and Pachos proposed an algorithm to implement this decomposition, but its implementation was experimental, contained mathematical limitations, and the original source code has since become unavailable \cite{SaEarp2005}. 
Meanwhile, the past decade has witnessed remarkable progress in the realization of quantum hardware with dozens of controllable qubits \cite{Sycamore}, along with more efficient techniques for implementing the output factors of Cartan decompositions in real devices \cite{Mansky_2023}.
Even with the fast development of hardware, the necessity for optimizing the use of contemporary hardware justifies a resurgence of interest in this construction among the quantum information community \cite{Liu_2021, PhysRevLett.129.070501}.

In this work, we revisit the original algorithm from \cite{SaEarp2005}, addressing its limitations and proposing a more reliable alternative. We also provide a Python implementation available in our GitHub repository \cite{github}. The original algorithm decomposes a general element $G\in\SU(2^n)$ into elements generated from certain maximal abelian subalgebras and two-qubit\footnote{It should be noted that the problem of decomposing from $\SU(4)$ down to $\SU(2)$ is already well-understood, and efficient algorithms for it are available \cite{PhysRevA.69.032315, QiskitTwoQubitBasisDecomposer}, which constitutes the natural initial step in our iteration scheme} unitaries in $\SU(4)^{\otimes n/2}$, by iteratively applying the so-called $KHK$ decomposition. The approach relies on numerically solving the zeros of a matrix polynomial derived by truncations of the Baker-Campbell-Hausdorff formula for products in a Lie group. 
However, when re-implementing this algorithm, we found two issues which led to improvements: (i) the inconvenience of constantly extracting matrix logarithms,  considering that exponential injectivity is only guaranteed in a neighborhood of the identity; and (ii) the uncertain convergence of the BCH series, so that truncation may not yield the expected result in some cases. 

To overcome these limitations, our novel approach consists of using the properties of involutive automorphisms of a Lie group, corresponding to an orthogonal symmetric Lie algebra $(\mathfrak{g},\theta)$ and its associated decomposition $\mathfrak{g}=\mathfrak{k}\oplus\mathfrak{m}$, to obtain a new method to calculate the elements $m\in \mathfrak{m}$ and $K\in\exp(\mathfrak{k})$ in the $G=K\exp(m)$ decomposition of any matrix $G\in \SU(2^n)$, for $n>2$. This gives us an algebraic process that replaces the numerical optimization required by the use of Baker-Campbell-Hausdorff expansions. As a bonus, the involutions almost completely eliminate the need to take logarithms of matrices, and they have the effect of concentrating the errors on the ‘Cartan’ components, which are not subsequently decomposed, thus avoiding error propagation.

Readers interested in a systematic treatment of Lie Theory might consult the excellent mathematical sources \cite{Helgason,Knapp1996,SanMartin2010, Gilmore_2008}.
\paragraph{Overview of the paper.}
In Section~\ref{sec: KHK decomp}, we review the theoretical foundation of the Cartan-Khaneja-Glaser (KHK) decomposition via symmetric Lie algebras and involutive automorphisms. Section~\ref{sec: Khaneja-Glaser bases} introduces the Khaneja-Glaser bases for $\su(2^n)$ and outlines the associated decomposition into maximally abelian subalgebras. Building on this, we develop in Section~\ref{sec: Breakdown of algorithm} a recursive algorithm to implement the KHK decomposition of any unitary matrix in $\SU(2^n)$, avoiding convergence and injectivity difficulties which arise in the original algorithm of \cite{SaEarp2005}, and culminating in the final factorization form of Corollary~\ref{def: cartann decomp}. Finally, some practical implementation issues,  error control strategies, performance benchmarking, and future developments are discussed in Section~\ref{sec: conclusion}. We validate our code with benchmarks in $\SU(8)$ and $\SU(16)$ and present a complete worked example in Appendix~\ref{app: example}. Additional technical details are provided in Appendices~\ref{appendix: phase}--\ref{app: BCH}.

\newpage
\section{ KHK decomposition and Khaneja-Glaser bases}
\label{sec: KHK and K-G decomp}

We will adopt the following notation throughout this paper
\begin{table}[h]
\centering
\begin{tabular}{lll}
$\mathbf{G}$   & capital bold    & Lie group or subgroup       \\
$G$            & capital         & group element           \\
$\mathfrak{g}$ & German Fraktur & Lie algebra or subspace, in particular $\fg=\Lie(\mathbf{G})$ \\
$\mathcal{G}$  & capital calligraphic   &  Lie algebra or subspace basis \\
$g$            & normal          & Lie algebra or subspace element        
\end{tabular}
\end{table}


\subsection{The KHK decomposition via involutive automorphisms}
\label{sec: KHK decomp}

\begin{definition}[{\cite[p.  229]{Helgason}}]
    \label{def: orthogonal sym Lie}
    Let $\mathfrak{g}$ be a Lie algebra over $\mathbb{R}$ and let $\theta$ an involutive automorphism of $\mathfrak{g}$, i.e. $\theta\neq I$ and $\theta^2=I$. If the set of fixed points of $\theta$ is a compactly embedded subalgebra of $\mathfrak{g}$, the pair $(\mathfrak{g},\theta)$ is called an \textbf{orthogonal symmetric Lie algebra}.  
\end{definition}

\begin{remark}[{\cite[p.  230]{Helgason}}]
    Let $\theta:\mathfrak{g}\rightarrow\mathfrak{g}$ be an involutive automorphism of a compact semisimple Lie algebra. Then:
\begin{enumerate}[(i)]
    \item $(\fg,\theta)$ is an orthogonal symmetric Lie algebra;

    \item 
    if $\fk\subset \fg$ is any $\theta$-invariant subspace, i.e. such that   $\theta(\mathfrak{k})=\mathfrak{k}$, then $\mathfrak{k}\cap \fZ_{\mathfrak{g}}=\{0\}$, where $\fZ_{\mathfrak{g}}$ denotes the center of $\mathfrak{g}$.
\end{enumerate}
\end{remark}

Let $(\mathfrak{g},\theta)$ be an orthogonal symmetric Lie algebra. Denote by $\mathfrak{k}$ and $\mathfrak{m}$ the eigenspaces of $\theta:\mathfrak{g}\rightarrow\mathfrak{g}$ for the eigenvalues $1$ and $-1$, respectively:

\begin{equation}
\label{eq: eigenv of theta}
 \theta(g)=
    \begin{cases}
        g & \text{if } g\in\mathfrak{k},\\
        -g & \text{if } g\in\mathfrak{m}.
    \end{cases}
\end{equation}
Then  $\mathfrak{g}=\mathfrak{k}\oplus\mathfrak{m}$ and 
\begin{equation}
    [\mathfrak{k},\mathfrak{k}] \subset \mathfrak{k}, \quad[\mathfrak{m},\mathfrak{k}]
    \subset \mathfrak{m}, \quad [\mathfrak{m},\mathfrak{m}]\subset \mathfrak{k}.
\end{equation}

Moreover, denoting by $B$ the bi-invariant Killing form on $\mathfrak{g}$, one has $\mathfrak{m}=\mathfrak{k}^{\perp}:=\fk^{\perp_B}$. 
The decomposition $\mathfrak{g}=\mathfrak{k}\oplus\mathfrak{m}$ is called the \textbf{symmetric decomposition of} $\mathfrak{g}$ associated to the \textbf{symmetric involution} $\theta$. If $B_{\theta}(g_1,g_2):=-B(g_1,\theta(g_2))$ is positive-definite, then the symmetric decomposition and symmetric involution are called the \textbf{Cartan decomposition} and \textbf{Cartan involution}, respectively.

\begin{definition}[{\cite[p.  209]{Helgason}}] 
    Let $(\mathfrak{g},\theta)$ be an orthogonal symmetric Lie algebra with eigenspaces \eqref{eq: eigenv of theta}, subordinate to a connected Lie group 
    $\mathbf{G}$ and a  distinguished Lie subgroup $\mathbf{K}\subset \mathbf{G}$ with Lie algebra $\mathfrak{k}$; then the pair $(\mathbf{G},\mathbf{K})$ is called the \textbf{symmetric pair} of $(\mathfrak{g},\theta)$. For convenience, we denote the corresponding symmetric decomposition  by $(\mathfrak{g}=\fk\oplus\fm,\theta)$.
\end{definition}

\begin{definition}[{\cite[p.  235]{Helgason}}]
    Given a symmetric decomposition  $(\mathfrak{g}=\fk\oplus\fm,\theta)$, we have a real algebra $\mathfrak{g}^*:=\mathfrak{k}\oplus i\mathfrak{m}$. The automorphism defined by 
\begin{equation*}
    \theta^*(k+im)
    :=k-im, 
    \qforq k\in\mathfrak{k}, \; m\in\mathfrak{m}, 
    \end{equation*}
    is an involution of $\mathfrak{g}^*$, and the pair $(\fg^*,\theta^*)=:(\fg,\theta)^*$ is an orthogonal symmetric Lie algebra, called the \textbf{dual} of $(\fg,\theta)$.
\end{definition}

\begin{lemma}
    Let $(\mathfrak{g}=\fk\oplus\fm,\theta)$ be a symmetric decomposition of the Lie algebra $\fg$. If $\fg$ is compact and semisimple, then $\fg^*$ is noncompact and semisimple. Additionally  $\fg^*=\mathfrak{k}\oplus i\mathfrak{m}$ is a Cartan decomposition of $\fg^*$. 
\end{lemma}
\begin{proof}
  The proof of the first part can be found in \cite[Proposition 2.1] {Helgason}. Now, $\fg$ and $\fg^*$ are real forms of the complexification $\fg^{\mathbb{C}}$, so their Killing forms are restrictions of the Killing form of $\fg^{\mathbb{C}}$. Then,   

  \begin{equation*}
      B_{\theta^*}(x+iy,x+iy)=-B(x+iy,x-iy)=-B(x,x)-B(y,y),
  \end{equation*}
  for all $x+iy\in \fk\oplus i\fm$. Since $\fg$ is compact and semisimple, $B$ is negative definite; therefore, $B_{\theta^*}$ is positive definite.  
\end{proof}

For a symmetric decomposition $(\mathfrak{g}=\mathfrak{k}\oplus\mathfrak{m},\theta)$ of $\fg$ wiht symmetric pair $(\mathbf{G},\mathbf{K})$, we assume $\mathbf{K}$ to be a closed and compact subgroup of $\mathbf{G}$. There is an involutive automorphism $\Theta:\mathbf{G}\rightarrow \mathbf{G}$ with differential $\theta$, such that  

\begin{equation}
\label{eq: Theta inv}
    \Theta|_{\mathbf{K}}=\mathrm{id},
\qquad
\Theta(\exp m)=\exp(-m), \quad \forall m\in\mathfrak{m}.
\end{equation}

\begin{remark}
    Suppose $\Theta:\mathbf{G}\rightarrow \mathbf{G}$ is an analytic involutive automorphism of a 
    compact simply connected Lie group. Then $\mathbf{K}$ is connected \cite[Theorem 8.2]{Helgason}, and the closed and compact subgroup $\mathbf{K}\subset \mathbf{G}$ is, in fact, totally spanned by the exponential map: $\mathbf{K}=\exp(\mathfrak{k})$.
\end{remark}

An interesting property is that, given a maximal Abelian subalgebra $\mathfrak{h}\subset\mathfrak{m}$, it is possible to obtain all of $\mathfrak{m}$ via the adjoint action of $\mathbf{K}$. The decomposition $\mathfrak{g}=\mathfrak{k}\oplus\mathfrak{m}$ of $\mathfrak{g}$ induces a decomposition in $\mathbf{G}$, which we call the KHK decomposition.

\begin{theorem}[KHK decomposition, {\cite[Theorem 6.7]{Helgason}}]
{\label{theo: KHK dcp}}
    Let  $(\fg=\fk\oplus \fm,\theta)$ be a Cartan decomposition of a semisimple Lie algebra, with symmetric pair $(\mathbf{G},\mathbf{K})$, and let $(\mathfrak{g}^*=\fk+i\fm,\theta^*)$ be the symmetric decomposition of the corresponding dual Lie algebra with symmetric pair $(\mathbf{G}^*,\mathbf{K}^*)$. Fixing a maximal Abelian subalgebra $\mathfrak{h}\subset \mathfrak{m}$, the following factorisation holds for any $G\in \mathbf{G}$:
\begin{equation*}
    G=K_0\exp(m) =K_0K_1\exp(h)K_{1}^{\dagger},
\end{equation*}
    where $K_0,K_1 \in \mathbf{K}$, $h\in\mathfrak{h}$, and $m=K_1hK_{1}^{\dagger}\in \mathfrak{m}$. Furthermore, for any $Q\in \mathbf{G}^*$, we have:
\begin{equation*}
    Q=K_2\exp(m_1) =K_2K_3\exp(h_1)K_{3}^{\dagger},
\end{equation*}
    where $K_2,K_3 \in \mathbf{K}^*$, $h_1\in\mathfrak{ih}$ and $m_1=K_3h_1K_{3}^{\dagger}\in \mathfrak{im}$.
\end{theorem}

\subsection{The Khaneja-Glaser special unitary bases}
\label{sec: Khaneja-Glaser bases}

From now on, we will focus on the Lie group $\mathbf{G}=\SU(2^n)$, whose Lie algebra $\mathfrak{g}=\su(2^n)$ consists of the traceless, skew-Hermitian matrices. To implement the KHK decomposition of $G\in \SU(2^n)$, we will make use of the Khaneja-Glaser bases of $\su(2^n)$, which are built upon the Pauli matrices

$$X=\begin{pmatrix}
    0 & 1\\
    1 & 0
\end{pmatrix}, \hspace{0.3cm} Y=\begin{pmatrix}
    0 & -i\\
    i & 0
\end{pmatrix}, \hspace{0.3cm} Z=\begin{pmatrix}
    1 & 0\\
    0 & -1
\end{pmatrix}.$$

\begin{definition} \label{def: KG basis}
    The \textbf{Khaneja-Glaser basis} for $\su(2^n)=\mathfrak{k}_n\oplus\mathfrak{m}_n$ with $\mathfrak{k}_n=\Span (\mathcal{K}_n)$ and $\mathfrak{m}_n=\Span (\mathcal{M}_n)$  is constructed recursively, as follows. Starting from
\begin{equation*}
    \mathcal{M}_2:=\frac{i}{2}\bigcup_{\alpha,\beta\in\{X, Y, Z\}}\{\alpha\otimes\beta\}, 
    \qandq 
    \mathcal{K}_2:=\frac{i}{2}\bigcup_{\alpha\in\{X, Y, Z\}}\{\alpha\otimes I, I\otimes \alpha\},
\end{equation*}
    we set    
\begin{equation*}
    \mathcal{G}_n
    :=\mathcal{M}_n\cup\mathcal{K}_n,
\end{equation*}
    and, by recursion,
\begin{align*}
    \mathcal{M}_n
    &=\left\{\frac{i}{2}I^{\otimes (n-1)}\otimes X,\frac{i}{2}I^{\otimes (n-1)}\otimes Y, \mathcal{G}_{n-1}\otimes X, \mathcal{G}_{n-1}\otimes Y\right\},\\
    \mathcal{K}_{n,0}
    &= \mathcal{G}_{n-1}\otimes I  \qandq  \mathcal{K}_{n,1}=\mathcal{G}_{n-1}\otimes Z,\\
    \mathcal{K}_n
    &=\left\{\frac{i}{2}I^{\otimes (n-1)}\otimes Z, \mathcal{K}_{n,0}, \mathcal{K}_{n,1}\right\}.
\end{align*}
    Similarly, we can define a maximal abelian subalgebra $\text{span}(\mathcal{H}_n)=\mathfrak{h}\subset\mathfrak{m}$ starting from
$\mathcal{H}_2:=\frac{i}{2}\bigcup_{\alpha\in\{X, Y, Z\}}\{\alpha\otimes\alpha\}$, set 
\begin{align*}
    \Bar{\mathcal{H}}_n
    &=\bigcup_{j=2}^{n-1}\left\{\mathcal{H}_j\otimes I^{\otimes(n-1-j)}\right\}
    \qandq 
    \mathcal{H}_n
    =\left\{\frac{i}{2}I^{\otimes(n-1)} \otimes X, \Bar{\mathcal{H}}_n\otimes X\right\},
\end{align*}
which will be used for the KHK decomposition later on.
\end{definition}

Note that
\begin{equation*}
    I^{\otimes k}=\underbrace{I\otimes \dots \otimes I}_{k},
\end{equation*}
and that, given a set $S$ and a matrix $A$, we denote by $S\otimes A$ the set obtained by taking the Kronecker product of each element in $S$ with the matrix $A$. 

In the case of $\su(4)$, we can visualize the Khaneja-Glaser basis and the corresponding subspaces in Figure \ref{fig: su4 basis}, and more generally, for $n>2$ qubits, in Figure \ref{fig: Khaneja-Glaser basis 2N}.
\begin{figure}[h]
    \centering
    \includegraphics[width=.5\linewidth]{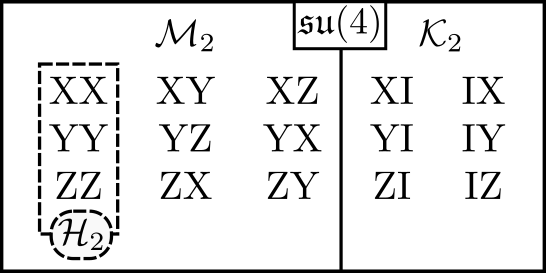}
    \caption{Khaneja-Glaser basis for $\su(4)$. We denote the element $\frac{i}{2}A\otimes B$ by $AB$ (adapted from \cite{SaEarp2005}).}
    \label{fig: su4 basis}
\end{figure}

\begin{figure}[h]
    \centering
    \includegraphics[width=\linewidth]{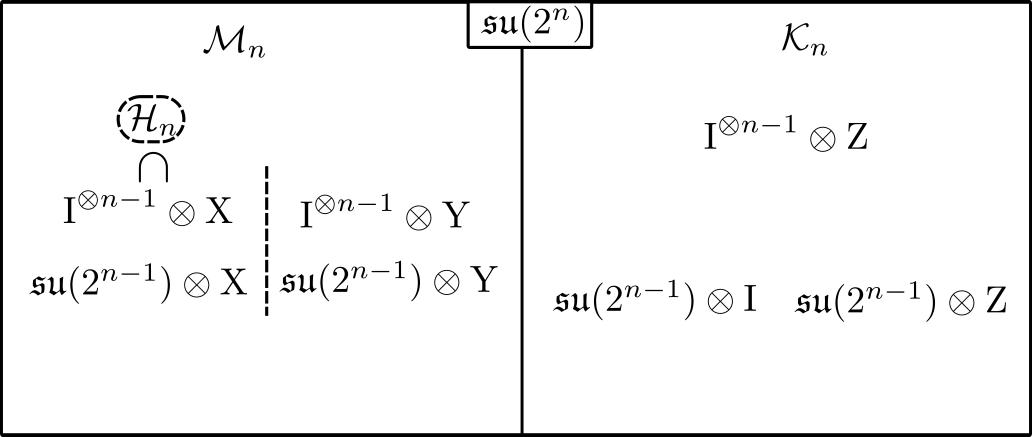}
    \caption{Khaneja-Glaser basis for $\su(2^n)$, again omitting the factor $\frac{i}{2}$ (adapted from \cite{SaEarp2005}).}
    \label{fig: Khaneja-Glaser basis 2N}
\end{figure}

To determine the KHK decomposition of some $G\in \SU(2^n)$, that is, to compute $K_0\in\mathbf{K}$ and $m\in\Span (\mathcal{M}_n)$ such that  $G=K_0\exp(m)$, we created an algorithm based on \cite{SaEarp2005} with elements from \cite{Nakajima2005ANA}.
The construction of the Khaneja-Glaser basis (Definition  \ref{def: KG basis}) induces the decomposition $\su(2^n)=\mathfrak{k}_n\oplus\mathfrak{m}_n$, associated with the involutive automorphism 
\begin{equation}
    \theta_Z(g)=(I^{\otimes n-1}\otimes Z)g(I^{\otimes n-1}\otimes Z),
\end{equation}
that is, for $n>2$, $(\su(2^n),\theta_Z)$ is an orthogonal symmetric Lie algebra (see Def.~\ref{def: orthogonal sym Lie}) and $\mathfrak{k}_n$ and $\mathfrak{m}_n$ are the eigenspaces of $\theta_Z$ for the $\pm1$ eigenvalues respectively. Furthermore, 
\begin{equation}
    \Theta_Z(G)=(I^{\otimes n-1}\otimes Z)G(I^{\otimes n-1}\otimes Z)
\end{equation}
is the involutive automorphism of Lie group $\SU(2^n)$ such that $d\Theta_Z=\theta_Z$ and that satisfies  \eqref{eq: Theta inv} (see \cite{Nakajima2005ANA}).

\begin{proposition}
\label{prop: m involution trick}
    Let $\mathbf{G}$ be a semisimple connected Lie group and $(\mathfrak{g},\theta)$ an orthogonal symmetric Lie algebra with symmetric pair $(\mathbf{G},\mathbf{K})$. Let $G\in\mathbf{G}$ such that $G=K_0\exp(m)$, as in Theorem \ref{theo: KHK dcp}, then 
\begin{equation}
\label{prop: m involution trick eq}
    \exp(2m)=\Theta(G^\dagger)G,
\end{equation}
    with $\Theta$ an involutive automorphism of $\mathbf{G}$ satisfying $d\Theta=\theta$ and  \eqref{eq: Theta inv}. 
\end{proposition}

\begin{proof}
    From  \eqref{eq: Theta inv}, we have    
    \begin{equation*}
        \Theta(G^\dagger)=\Theta(\exp(-m)K_{0}^\dagger)=\Theta(\exp(-m))\Theta(K_{0}^\dagger)=\exp(m)K_{0}^\dagger.
    \end{equation*}
Therefore,
\begin{align*}
    \Theta(G^\dagger)G
    &=(\exp(m)K_{0}^\dagger)(K_0\exp(m))\\
    &=\exp(m)^2.
    \qedhere
\end{align*}
\end{proof}

In principle, Proposition \ref{prop: m involution trick} offers a practical method for determining the element $m\in\mathfrak{m}$ in the KHK decomposition (Thm.~\ref{theo: KHK dcp}). However, in practice, numerically computing the logarithm of an $\SU(2^n)$ matrix involves diagonalization and the logarithm of each eigenvalue \cite{Liu_2021}, for which may not be a unique answer. In general, this process can lead to the log of the matrix lying outside the algebra\footnote{Notice, for example, that when calculating the log of $-I$, the most natural answer would have $\tr(\log(-I))\neq 0$.}. 

\begin{proposition}
Let $(\su(2^n),\theta_U)$ be an orthogonal symmetric Lie algebra, with $U \in \mathrm{SU}(2^n)$ and
\[
\theta_U(g) := \operatorname{Ad}_U(g).
\]
Then
\begin{equation}
m = \frac{1}{2}\log\!\bigl(\Theta(G^{*})\,G\bigr)
\label{eq:m-solution}
\end{equation}
solves~\eqref{prop: m involution trick eq} in Proposition~\ref{prop: m involution trick} under the assumption that $G$ has no negative eigenvalues.
\end{proposition}

\begin{proof}
It suffices to show that if $m\in\mathfrak{m}$,, then $\log(e^m)\in\mathfrak{m}$. Let us observe that
\begin{equation*}
    \log(G)^{\dagger}=\log(G^{\dagger})
    =-\log(G),
    \quad\forall G\in \SU(2^n), 
\end{equation*} 
hence $\log(G) \in \mathfrak{u}(2^n)$. As
\[
\mathfrak{u}(2^n) = \mathfrak{su}(2^n) \oplus \mathfrak{u}(1),
\]
there exist $g \in \mathfrak{su}(2n)$ and $\varphi \in \mathbb{R}$ such that
\[
\log(G) = g + i\varphi I.
\]

 Now, 
\begin{equation*}
    \theta_U(\log(e^m))
    =U\log(e^m)U^\dagger=\log(\Theta(e^m))=-\log(e^m).
\end{equation*}
    Since $\theta_{U}(i\varphi I)=i\varphi I$, then $\log(e^m)\in\mathfrak{m}$.
\end{proof}
\begin{remark}
In practice, if the spectrum of $\Theta(G^\dagger)G$ contains no negative real eigenvalues, we compute the matrix logarithm using the principal branch.
If $\Theta(G^\dagger)G$ admits negative eigenvalues, the principal logarithm is not well-defined, and we instead calculate the logarithm with a brute force optimization, to guarantee it lands in the correct branch of the algebra.
Let $\mathcal{G}=\{G_1,\dots,G_p\}$ be a basis of the target algebra and denote by $\mathrm{span}(\mathcal{G})$ its linear span. We then define
\[
\log(G)\;:=\;\arg\min_{g\in \mathrm{span}(\mathcal{G})}
\bigl\|\,G-\exp(g)\bigr\|_F .
\]
Equivalently, writing $g=\sum_{i=1}^p a_i G_i$, we optimize over $a\in\mathbb{R}^p$.
\end{remark}


Obtaining $m\in\mathfrak{m}_n$ in the decomposition $G=K_0\exp(m)$ of $G\in \SU(2^n)$, we then find $K_0=G\exp(-m)$. To complete the KHK decomposition we use the strategy proposed by in \cite{SaEarp2005} to determine $K_1\in \mathbf{K}$ and $h$ in the maximal Abelian subalgebra $\mathfrak{h}_n=\Span (\mathcal{H}_n)$, such that $m=K_1hK_{1}^\dagger$. However, this does not yet allow for the recursive implementation of our decomposition, since we haven't produced and element in $\su(2^{n-1})$. For that we will need to introduce a few new subspaces and a subsequent KHK decomposition, cf. Figure \ref{fig:fn construction}.
\begin{figure}[h]
    \centering

\begin{framed}
    \begin{equation}
         \mathcal{K}_{n,0}
         =\mathcal{G}_{n-1}\otimes I, \quad \mathcal{K}_{n,1}
         =\mathcal{G}_{n-1}\otimes Z,
    \end{equation}
    \begin{equation}
        \mathcal{F}_2
        =\{0\}, \quad \mathcal{F}_n
        =\left\{\left\{\bigcup_{j=2}^{n-1} \mathcal{H}_{j}\otimes I^{\otimes(n-1-j)}\right\} \otimes Z\right\}
    \end{equation}
\end{framed}
    \caption{Construction of the maximal Abelian subalgebra basis}
    \label{fig:fn construction}
\end{figure}

\begin{figure}[h]
    \centering
    \includegraphics[width=\linewidth]{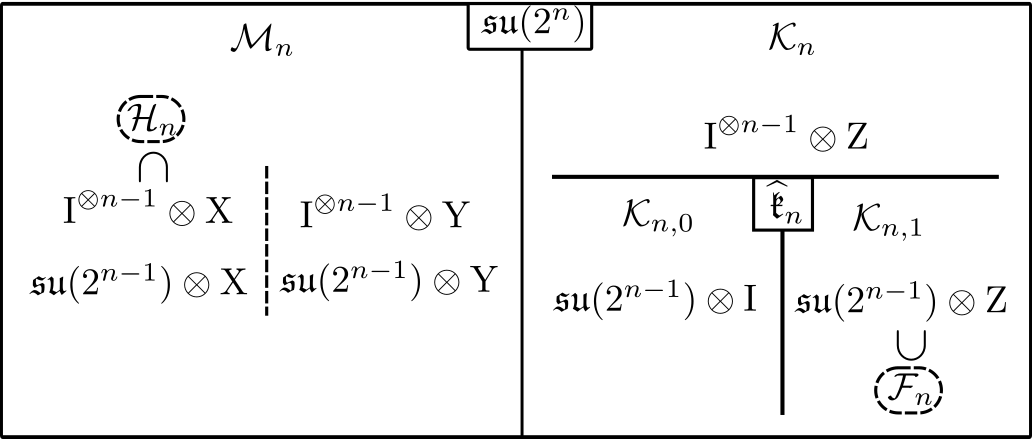}
    \caption{Decomposition of $\su(2^n)$ into both Cartan pairs. The elements $I^{\otimes n-1}\otimes A$, $A=X, Y, Z$ are multiplied by $\frac{i}{2}$ (adapted from \cite{SaEarp2005}).}
    \label{fig:fn algebra}
\end{figure}

Notice how, for $n>2$, we can decompose the  subalgebra $\mathfrak{k}_n$ into
\begin{align*}
\begin{split}
    \mathfrak{k}_n
    &=\Span (\mathcal{K}_{n,0})\oplus \Span (\mathcal{K}_{n,1})\oplus \Span \left(\frac{i}{2}I^{\otimes(n-1)}\otimes Z\right)\\
    &\simeq \su(2^{n-1})\oplus\su(2^{n-1})\oplus \mathfrak{u(1)}.
\end{split}
\end{align*}

Let  $\widehat{\mathfrak{k}}_n:=\Span (\mathcal{K}_{n,0})\oplus \Span (\mathcal{K}_{n,1})$. Since $\left[I^{\otimes(n-1)}\otimes Z,\mathfrak{k}_n\right]=0$, and $\mathbf{K}=\exp(\mathfrak{k}_n)$, for any $G\in \SU(2^n)$ we have 
\begin{align*}
\begin{split}
    G&=K_0K_1\exp(h)K_1^{\dagger}\\
    &= \widehat{K_0K_1}\left(I^{\otimes n-1}\otimes\exp\left(\frac{i\alpha}{2}Z\right)\right) \exp(h)\widehat{K_1^{\dagger}}\left(I^{\otimes n-1} \otimes\exp\left(\frac{i\beta}{2}Z\right)\right),
\end{split}
\end{align*}
where $\alpha,\beta\in \mathbb{R}$ and $\widehat{K_0K_1}, \widehat{K_1^{\dagger}}\in exp(\widehat{\mathfrak{k}}_n$).

These new subspaces open the way for a secondary decomposition, using a different involution proposed by \cite{Drury_2008}, albeit they were applying it to a different maximal Abelian subalgebra. Letting
\begin{equation}
    \theta_X(g):=(I^{\otimes n-1}\otimes X)g(I^{\otimes n-1}\otimes X), \qforq 
    g\in \widehat{\mathfrak{k}}_n,
\end{equation}
$(\widehat{\mathfrak{k}}_n, \theta_X)$ is an orthogonal symmetric Lie algebra with decomposition $\widehat{\mathfrak{k}}_n=\mathfrak{k}_{n,0}\oplus \mathfrak{k}_{n,1}$, where 
$$
\mathfrak{k}_{n,0}=\Span (\mathcal{K}_{n,0})
\qandq
\mathfrak{k}_{n,1}=\Span (\mathcal{K}_{n,1}). 
$$ 

We can visualize these new subspaces in Figure \ref{fig:fn algebra}. Now take the maximal Abelian subalgebra $\mathfrak{f}_n:=\Span (\mathcal{F}_{n})\subset \mathfrak{k}_{n,1} $, as constructed in Figure \ref{fig:fn construction}. These choices allow one to decompose the elements $\widehat{K_0K_1}$ and $\widehat{K_1^{\dagger}}$, using Proposition \ref{prop: m involution trick} and the fact that 
\begin{equation}
    \Theta_X(G)=(I^{\otimes n-1}\otimes X)G(I^{\otimes n-1}\otimes X)
\end{equation}
is an involutive automorphism of the  compact semisimple simply connected Lie group whose Lie algebra is $\widehat{\mathfrak{k}}_n$, satisfying $d\Theta_X=\theta_X$ and \eqref{eq: Theta inv}. We illustrate the new subspace, alongside the spaces selected by the two involutions for the case of $\mathfrak{su}(8)$ in Fig.~\ref{fig:su8}.

\begin{figure}
    \centering
    \includegraphics[width=\linewidth]{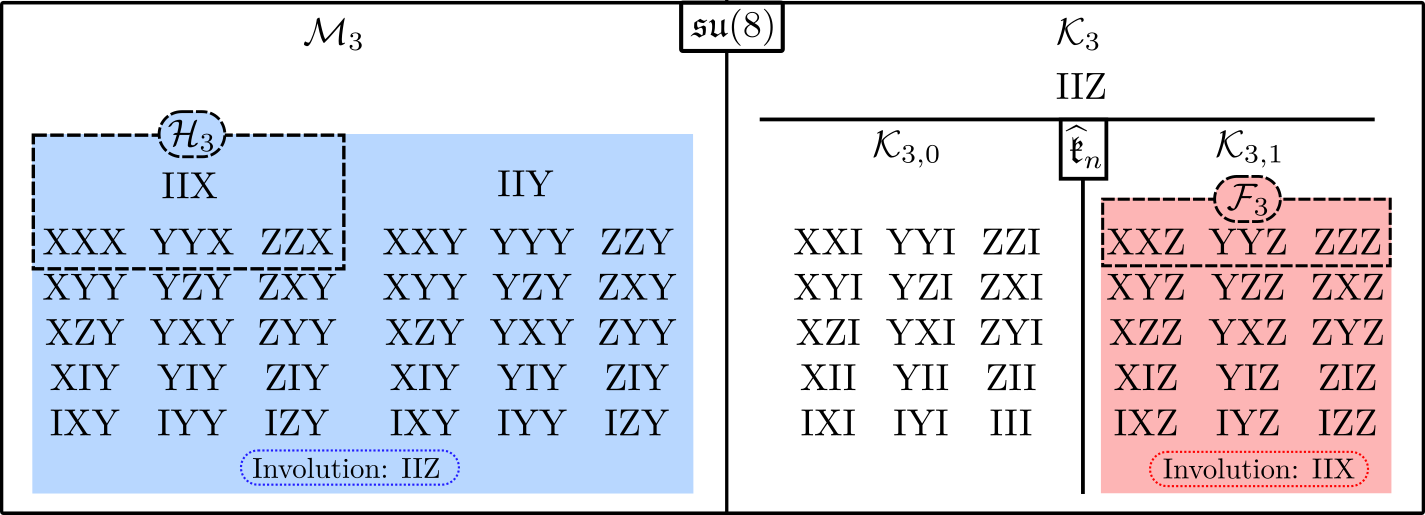}
    \caption{Khaneja-Glaser basis for $\su(8)$, again omitting the factor $\frac{i}{2}$. We also indicate the subspaces that are selected by the involutions we introduce in this paper.}
    \label{fig:su8}
\end{figure}

\begin{remark}
     Any $K\in\exp(\mathfrak{k}_n)$ can be factored as   $K=\widehat{K}\widetilde{K}$, with $\widehat{K}\in\exp(\widehat{\mathfrak{k}}_n)$ and 
     $$\widetilde{K}\in\exp\left(\Span \left(\frac{i}{2}I^{\otimes(n-1)}\otimes Z\right)\right).$$
     Furthermore, $K=K_{1,0}\exp(\widehat{m})\widetilde{K}$, where $K_{1,0}\in \exp(\mathfrak{k}_{n,0})$ and $\widehat{m}\in\mathfrak{k}_{n,1}$. Observe that $\Theta_X(\widetilde{K})=\widetilde{K}^\dagger$, thus 
\begin{equation}
\label{remark: involution and phase}
    \left(\exp(\widehat{m}\right) \widetilde{K})^2 =\exp(2m_Z)  =\Theta_X(K^\dagger)K,
\end{equation}
    with $m_Z\in \mathfrak{k}_{n,1}\oplus\Span (\frac{i}{2}I^{\otimes(n-1)}\otimes Z)$.
\end{remark}

\begin{remark}
    Analogously to equation \eqref{eq:m-solution},  
    $m_Z=\frac{1}{2}\log(\Theta_X(K^\dagger)K)$ 
    satisfies \eqref{remark: involution and phase}.
\end{remark}

Implementing a $\mathbf{KHK}$ decomposition on $\widehat{K_0K_1}$ and $\widehat{K_{1}^\dagger}$ yields

\begin{equation}
    \widehat{K_0K_1}= K_{1,0}K_{1,1}\exp\left(f^{(0)}\right)K_{1,1}^{\dagger},
\end{equation}
\begin{equation}
    \widehat{K_1^{\dagger}}= K_{2,0}K_{2,1}\exp\left(f^{(1)}\right)K_{2,1}^{\dagger},
\end{equation}
with $f^{(0)},f^{(1)}\in \mathfrak{f}_n$ and $K_{i,j}\in \SU(2^{n-1})\otimes SU(2)$. The $\SU(2^{n-1})$ matrices can be then obtained by removing the even rows and columns, taking into consideration a possible phase correction as described in Appendix~\ref{appendix: phase}.

Building upon the discussions in this section, we now obtain the Khaneja-Glaser decomposition as a consequence of  Theorem \ref{theo: KHK dcp}:

\begin{corollary}[Khaneja-Glaser decomposition]
\label{def: cartann decomp}
    Any $G\in \SU(2^n)$ can be factored as
\begin{equation*}
    G=K^{(0)}\otimes I \cdot e^{f^{(0)}} \cdot K^{(1)}\otimes I \cdot I^{\otimes (n-1)}\otimes\widetilde{K}^{(0)} \cdot  e^{h^{(0)}} \cdot K^{(2)}\otimes I\cdot  e^{f^{(1)}} \cdot K^{(3)}\otimes I \cdot I^{\otimes (n-1)}\otimes\widetilde{K}^{(1)} 
\end{equation*}
    for some $K^{(i)}\in \SU(2^{n-1})$, $\widetilde{K}^{(i)}\in \SU(2)$, $h^{(0)}\in \mathfrak{h}_n$, $f^{(i)}\in \mathfrak{f}_n$.    
\end{corollary}

Observe that this decomposition can be applied recursively to the $K^{(i)}$ factors, down to elements of $\SU(4)$ and Abelian factors.

\section{Breakdown of the algorithm: step-by-step summary}
\label{sec: Breakdown of algorithm}

Our algorithm recursively implements the Khaneja-Glaser decomposition, factoring a large special unitary matrix to the level of single qubit unitaries and Abelian factors. In this section, we will describe its implementation scheme. A complete example of its application is illustrated in Appendix~\ref{app: example}.

    Given $G\in \SU(2^n)$, we implement its decomposition by following these steps:
\begin{enumerate}
    \item Define the involutive automorphism \label{item: comp. basis}
    \begin{equation*}
        \Theta_Z(G)
        :=(I^{\otimes n-1}\otimes Z) G (I^{\otimes n-1}\otimes Z).
    \end{equation*}

    \item Generate the Khaneja-Glaser basis for $\SU(2^n)$. 
            
    \item Compute the subspace element
    \begin{equation*}
        m_0=\frac{1}{2}\log(\Theta_Z(G^\dagger)G)\in\mathfrak{m}_n.
    \end{equation*}\label{item: m calc}
          
    \item Compute \label{item: log m}
    $K_{0,0}=G\exp(-m_0)$.
            
    \item Alphabetically order the basis elements of $\mathfrak{h}_n$ and define \label{item: v def}
    $$v=\sum_{u_i\in\mathcal{H}_n}\pi^{i-1}u_i,$$ 
    which generates a dense 1-parameter subgroup in  $\mathbf{H}$ \cite{SaEarp2005}.
           
    \item From $v$ and $m_0$, define the function 
    $$f_{v,m_0}(K)=\langle v, \Ad_{K}(m_0)\rangle,$$
    where $\langle a,b \rangle=\tr(\text{ad}_a,\text{ad}_b)$ is the Killing form on $\su(2^n)$.
                      
    \item  Find $K_{0,1}$, by minimizing \label{item: opt 2}
    $$K_{0,1} =\min_{k\in \mathfrak{k}_n} f_{v,m_0}(\exp(k)),$$
    over $\mathcal{K}_n$ \cite{SaEarp2005}.

    \item Compute the Abelian factor\label{item: h}
    $$h^{(0)}=K_{0,1}^\dagger m_0K_{0,1}.$$
    
    \item Assemble the KHK decomposition of $G$\label{item: decomp 1}:
    \begin{equation*}
        G= K_{0,0}K_{0,1}\exp(h^{(0)})K_{0,1}^{\dagger}.
    \end{equation*}

    \item Define the involutive automorphism  
    $$\Theta_X(G):=(I^{\otimes n-1}\otimes X) G (I^{\otimes n-1}\otimes X).$$

    \item \label{compute m_1 m_2}Compute the subspace elements $m_1,m_2\in \mathfrak{k}_{n,1} \oplus \Span (\frac{i}{2}I^{\otimes(n-1)}\otimes Z)$:
    \begin{align*}
        m_1
        &=\frac{1}{2} \log(\Theta_X(K_{0,1}^{\dagger}K_{0,0}^{\dagger})K_{0,0}K_{0,1}) \in\mathfrak{k}_{n,1} \oplus \Span (\frac{i}{2}I^{\otimes(n-1)}\otimes Z),\\
        m_2&=\frac{1}{2} \log(\Theta_X(K_{0,1})K_{0,1}^{\dagger})\in\mathfrak{k}_{n,1}\oplus \Span (\frac{i}{2}I^{\otimes(n-1)}\otimes Z).
    \end{align*}\label{item: m_2 calc}
    
    \item Compute 
    $$K_{1,0} =K_{0,0}K_{0,1}\exp(-m_1) \qandq K_{2,0}=K_{0,1}^\dagger\exp(-m_2);$$ \label{item: K_1 calc}

    \item Separate the $I^{\otimes(n-1)}\otimes Z$ phase from the factors
    \begin{equation*}
        \widehat{m}_1 =\mathrm{proj}_{\mathfrak{k}_{n,1}}(m_1), \hspace{0.5cm}\hspace{0.5cm} \widehat{m}_2=\mathrm{proj}_{\mathfrak{k}_{n,1}}(m_2),
        \end{equation*}
        \begin{equation*}
            \widetilde{m}_1=m_1-\widehat{m}_1 , \hspace{0.5cm}\hspace{0.5cm} \widetilde{m}_2=m_2-\widehat{m}_2;
        \end{equation*} \label{item: z phase}

    \item \label{item: seccond decomp} Repeat steps \ref{item: v def}-\ref{item: h} for $\widehat{m}_1$ and $\widehat{m}_2$, replacing $\mathcal{K}_n\rightarrow\mathcal{K}_{n,0}$, and $\mathfrak{h}_n\rightarrow\mathfrak{f}_{n}$, to get 
    \begin{equation*}
        \widehat{m}_1 = K_{1,1}f^{(0)}K_{1,1}^\dagger
        \qandq
        \widehat{m}_2 =K_{2,1}f^{(1)}K_{2,1}^\dagger;
    \end{equation*}

    \item Assemble the KHK decomposition of $K_{0,0}K_{0,1}$ and $K_{0,1}^\dagger$:
    \begin{align*}
        K_{0,0}K_{0,1} 
        &= \exp(\widetilde{m}_1) K_{1,0} K_{1,1} \exp(f^{(0)})K_{1,1}^{\dagger},\\
        K_{0,1}^\dagger
        &=\exp(\widetilde{m}_2) K_{2,0} K_{2,1} \exp(f^{(1)}) K_{2,1}^{\dagger},
    \end{align*}
    where $f^{(0)},f^{(1)}\in \mathfrak{f}_n$,  $\widetilde{m}_1, \widetilde{m}_2\in I^{\otimes(n-1)}\otimes \su(2) $ and $K_{i,j}\in \SU(2^{n-1})\otimes I$.

    \item Construct the elements of $ SU(2^{n-1})$ (see Appendix~\ref{appendix: phase}):
        $$\varphi_1 =\frac{1}{2^{n-1}}\arg  (\det(K'_{1,0}K'_{1,1})),
        \quad 
        \varphi_2=\frac{1}{2^{n-1}}\arg (\det(K'_{2,0}K'_{2,1})),$$
        $$K^{(0)}=e^{-i\varphi_1}K'_{1,0}K'_{1,1}\quad K^{(1)}=(K'_{1,1})^{\dagger},$$
        $$K^{(2)}=e^{-i\varphi_2}K'_{2,0}K'_{2,1}\quad K^{(3)}=(K'_{2,1})^{\dagger},$$
    where each $K'_{i,j}$ is obtained by removing the even rows and columns of $K_{i,j}$.\label{item: su(n-1)}

    \item Construct elements in $\SU(2)$ from $\widetilde{m}_1$ and $\widetilde{m}_2$:
    \[ \widetilde{K}^{(0)} =\exp(\widetilde{m}'_1)
    \qandq \widetilde{K}^{(1)}
    =\exp(\widetilde{m}'_2),\]
    where each $\widetilde{m}'_i$ is obtained by removing all but the first two columns and rows of $\widetilde{m}'_2$.

    \item Assemble the Khaneja-Glaser decomposition of $G$   \label{item: decomp 2}
    \begin{equation*}
        G =e^{i\varphi}K^{(0)}\otimes I \cdot e^{f^{(0)}} \cdot K^{(1)}\otimes I \cdot I^{\otimes (n-1)} \otimes\widetilde{K}^{(0)} \cdot  e^{h^{(0)}} \cdot K^{(2)}\otimes I\cdot  e^{f^{(1)}} \cdot K^{(3)}\otimes I \cdot I^{\otimes (n-1)} \otimes\widetilde{K}^{(1)}, 
    \end{equation*}
    where $\varphi=\varphi_1+\varphi_2$, $K^{(i)}\in \SU(2^{n-1})$, $\widetilde{K}^{(i)}\in \SU(2)$, $h^{(0)}\in \mathfrak{h}_n$ and $f^{(i)}\in \mathfrak{f}_n$.
                
    \item Repeat steps \ref{item: comp. basis}-\ref{item: decomp 2}, iterating $n\rightarrow n-1\geq 3$ and substituting $G$ with $K^{(i)}$, aggregating the $\varphi_i$ phases and relabeling the superscripts after each decomposition; \label{item: reccursive decomp}

    \end{enumerate}

\begin{remark}
In finite precision, Step~\ref{item: v def} uses a truncated vector $\widehat v=v+\delta v$. Since $f_{v,m_0}(K)=\langle v,\Ad_K(m_0)\rangle$ is linear in $v$, the objective perturbation satisfies
\begin{equation}
    \bigl|f_{\widehat v,m_0}(K)-f_{v,m_0}(K)\bigr|
=\bigl|\langle \delta v,\Ad_K(m_0)\rangle\bigr|
\le \|\delta v\|\,\|m_0\|\qquad(\forall K),
\end{equation}
so the minimization in Step~\ref{item: opt 2} is affected only at first order in $\|\delta v\|$.

Moreover, writing the computed minimizer as $\widehat K_{0,1}=K_{0,1}+\Delta K$, the induced perturbation on Step~\ref{item: h} obeys the exact expansion
\begin{equation}
   \widehat h^{(0)}-\;h^{(0)}
=(K_{0,1}+\Delta K)^\dagger m_0 (K_{0,1}+\Delta K)-K_{0,1}^\dagger m_0 K_{0,1}
=K_{0,1}^\dagger m_0\Delta K+\Delta K^\dagger m_0 K_{0,1}+\Delta K^\dagger m_0\Delta K, 
\end{equation}
hence
\begin{equation}
   \|\widehat h^{(0)}-h^{(0)}\|\;\le\;2\|m_0\|\,\|\Delta K\|+\|m_0\|\,\|\Delta K\|^2
=\;O(\varepsilon_{\mathrm{mach}})\,\|m_0\|. 
\end{equation}
Therefore, even under the pessimistic assumption that the entire perturbation lies in $h_n^\perp$, its norm remains at the machine-precision scale.
\end{remark}

\section{Implementation analysis and benchmarking}
\label{sec: conclusion}

We implemented the algorithm in Python, in a jupyter notebook, and the code is available at \cite{github}.
To implement step \ref{item: opt 2}, we used the optimization package \enquote{scipy.optimize}, invoking the functions \enquote{root} and \enquote{minimize} respectively. Whenever the algorithm completes an iteration of steps \ref{item: comp. basis}-\ref{item: decomp 2}, it stores the factors in a list through a custom Python class that preserves the product order. Finally, it labels every factor with a superscript.

\subsection{Error management}
Analyzing the precision of our decomposition, requires considering two different types of error. The \emph{approximation error} measures how well the factors approximate the original matrix, in the matrix norm:
\begin{equation*}
    E_{a}(G):=\|G-\Bar{G}\|.
\end{equation*}

The \emph{subspace error} is associated to factors not being quite in the correct subspace. For $h\in \mathfrak{h}_j$, we define
\begin{equation}
    E_{s}(h)=\frac{1}{m}\|([h,h_1],...,[h,h_m])_{h_i \in \mathcal{H}_j}\|,
\end{equation}
and analogously for $f\in \mathfrak{f}_j$. Note that, by construction, only the computations of Abelian factors incur in subspace errors. 
\begin{remark}
    Another drawback to the approach in \cite{SaEarp2005} is that the it produces subspace errors in the $K$ elements, which will propagate in the subsequent decompositions, since the root of the polynomial in \eqref{eq: BCH polynomial}, resulting from truncations of the BCH expansion (see Appendix~\ref{app: BCH}), is just an approximation of the matrix $m$ that satisfies \eqref{eq: BCH basis}, hence may not actually belong to the subspace $\mathfrak{m}_n$. We can easily solve this problem by taking the projection onto $\mathfrak{m}_n$, however doing so implies, in some cases, increasing the approximation error. 
    Our implementation circumvents this problem by only incurring in subspace errors in the $F$ and $H$ elements.
\end{remark}

We tested our code by generating $10000$ random matrices from $\SU(8)$ and $500$ from $\SU(16)$ using the \enquote{scipy.stats.unitary\_group} function, and applying the decomposition on a Google Collab notebook. We display the error associated with these decompositions in Table~\ref{tab: results}. 

\begin{table}[h]
\centering
\begin{tabular}{l|lllll}
       & Time (s)  & Mean $E_{a}$ &  Mean $E_{s}$ & $\sigma(E_{s})$ \\ \hline
SU(8)  & $0.9$       & $2.2\cdot10^{-14}$                            & $2.3\cdot10^{-6}$                                 & $1.7\cdot10^{-6}     $                                                     \\
SU(16) & $256.7$       & $1.2\cdot10^{-13}$                             & $5.7\cdot10^{-5}$                                 & $4.6\cdot10^{-4}  $                                               
\end{tabular}
\caption{Benchmarks for the algorithm for $\SU(8)$ and $\SU(16)$}
\label{tab: results}
\end{table}


\subsection{Comparison with existing approaches}
\label{app: Analysis}

There are many papers involving the Cartan decomposition of $\SU(2^n)$. In fact, this work references several results in this direction. Among them, \cite{SaEarp2005} presents an explicit algorithm for this decomposition using the Khaneja-Glaser basis, which serves as a foundation for our work. 

When using the Cartan decomposition to factor unitary matrices, one of the main challenges one must overcome is how to deal with the transition between the group and its algebra — which is mediated by the (a priori non-injective) exponential map. Although \cite{SaEarp2005} appears to resolve this problem using mathematically rigorous methods, their approach exhibits two critical limitations:

\begin{itemize}
    \item [(i)] To determine $K_0$ and $m$ satisfying the equation $G = K_0\exp(m)$, \cite{SaEarp2005} employs the Baker – Campbell – Hausdorff (BCH) expansion (see Appendix~\ref{app: BCH}). We encountered several challenges in implementing this approach. The main issue is that the series is only guaranteed to converge for certain elements of the algebra \cite{Stefano}. Consequently, there is no guarantee that the polynomial root of \eqref{eq: BCH polynomial} correctly corresponds to the target matrix $m$.
    
    \item [(ii)] The algorithm proposed in \cite{SaEarp2005} also does not explicitly describe how to compute the $\log$ of arbitrary elements of $\SU(2^n)$. This poses a problem because, in general, the logarithm of a special unitary matrix is not unique, introducing significant errors in the KHK decomposition.
\end{itemize}

In this work, we overcome these issues by leveraging involutive automorphisms to minimize the need for logarithm computations, and we observe that when such computations are necessary, they can be performed in a controlled manner, either directly or through a brute force optimization, always ensuring the result remains within the desired subspace. We also provide a Python implementation of the algorithm, which was not available for \cite{SaEarp2005}. Furthermore, while the use of involutions was inspired by \cite{Nakajima2005ANA,Drury_2008}, our approach improves upon these algorithms: after decomposing an $n$-qubit unitary, we can recursively decompose the resulting $(n-1)$-qubit unitaries. We also use a different maximal Abelian subalgebra, the one generated by Khaneja-Glaser basis. As a result, the Cartan factors produced by our decomposition have a near-optimal implementation into C-NOTs and single qubit rotations, which was recently proposed \cite{Mansky_2023}. 

In conclusion, although our work builds upon the algorithm developed in \cite{SaEarp2005} and incorporates techniques from \cite{Nakajima2005ANA} and \cite{Drury_2008}, neither of these perspectives alone is sufficient to obtain a decomposition for which every resulting factor has a near-optimal computational implementation.

\subsection{Conclusion and future developments}

In this study, we have successfully revisited and improved the algorithm from \cite{SaEarp2005} for the Cartan-Khaneja-Glaser decomposition of $\SU(2^n)$. By leveraging the structure of symmetric Lie algebras and their associated involutive automorphisms, we developed an algebraic decomposition procedure that bypasses several limitations of the original method—most notably, the reliance on truncated Baker–Campbell–Hausdorff series and the repeated use of matrix logarithms, which are often ill-defined or numerically unstable in practice. 
The method naturally leads to a recursive decomposition of multi-qubit unitaries down to abelian components and $\SU(2)$ evolutions, for which efficient implementations are well-known. We also show that errors introduced during decomposition are confined to Abelian factors, ensuring that they do not propagate throughout the circuit—a critical feature for applications in fault-tolerant quantum computing. We implemented our algorithm in Python, making it available in an open-source repository, and validated its effectiveness using thousands of randomly generated matrices in $\SU(8)$ and $\SU(16)$.

Nevertheless, our algorithm has room for improvement, particularly in its scalability with the increasing number of qubits. One immediate improvement would be the parallelization of the decomposition in steps \ref{item: seccond decomp} and \ref{item: reccursive decomp}. The most demanding computational task, step \ref{item: opt 2}, could also be improved if optimized for GPU processing, or by introducing an efficient method of gradient calculation.
Another significant enhancement to our algorithm would involve bypassing the optimization process outlined in step \ref{item: opt 2}, aiming for the direct computation of variables $h$ and $K_1$. Implementing this could substantially improve scalability, broadening its utility to more complex quantum systems. In that regard, further research should consider whether the approach presented in \cite{Nakajima2005ANA} could be effectively adapted to our decomposition, with the Khaneja-Glaser basis.
Finally, our code could be appended with a method to decompose the $\SU(4)$ unitaries into single qubit $\SU(2)$ elements, such as the ones proposed in \cite{PhysRevA.69.032315, QiskitTwoQubitBasisDecomposer}. Finally, the approach  proposed by \cite{vatan2004realization}, and subsequently expanded by \cite{Mansky_2023}, could be implemented to decompose the Abelian factors into SWAP and C-NOT gates, which are beneficial for numerous applications.

\section*{Acknoledgments}
JMR was partially supported by the Coordenação de Aperfeiçoamento de Pessoal de Nível Superior (CAPES)- Finance Code 001 
AD was partially supported by the Coordenação de Aperfeiçoamento de Pessoal de Nível Superior (CAPES)- Finance Code 001 and the São Paulo Research Foundation (Fapesp) \mbox{[2024/23590-2]}. HSE was supported by the São Paulo Research Foundation (Fapesp)    \mbox{[2021/04065-6]} \emph{BRIDGES collaboration} and the Brazilian National Council for Scientific and Technological Development (CNPq)  \mbox{[311128/2020-3]}. MTC was partially supported by CNPq grant \mbox{[311314/2023-6]} and the São Paulo Research Foundation (Fapesp) \mbox{}{[2024/16657-3]}.

The method described in this paper was submited for patent protection.

\newpage
\printbibliography

\appendix

\newpage

\section{Detailed example}
\label{app: example}
We will now demonstrate the decomposition with the $\SU(8)$ matrix:
\[\hspace*{-1cm} 
G=\begin{pmatrix}
- \frac{84}{229} + \frac{131}{844}i & - \frac{249}{962} - \frac{12}{437}i & \frac{287}{928} - \frac{2}{751}i & - \frac{110}{651} - \frac{9}{770}i & - \frac{179}{552} + \frac{38}{173}i & - \frac{110}{651} - \frac{9}{770}i & - \frac{310}{771} + \frac{59}{121}i & \frac{249}{962} + \frac{12}{437}i \\
\frac{175}{883} + \frac{131}{638}i & \frac{185}{774} + \frac{28}{313}i & \frac{207}{953} + \frac{48}{191}i & \frac{39}{139} + \frac{67}{492}i & \frac{207}{953} + \frac{48}{191}i & \frac{18}{427} - \frac{357}{860}i & - \frac{175}{883} - \frac{131}{638}i & \frac{466}{927} - \frac{154}{933}i \\
\frac{47}{578} + \frac{109}{877}i & \frac{124}{993} + \frac{8}{475}i & - \frac{49}{430} + \frac{427}{758}i & - \frac{96}{179} & \frac{65}{571} + \frac{118}{941}i & - \frac{96}{179} & \frac{7}{68} - \frac{22}{609}i & - \frac{124}{993} - \frac{8}{475}i \\
- \frac{1}{38} + \frac{182}{989}i & \frac{67}{263} - \frac{112}{421}i & - \frac{69}{374} - \frac{251}{741}i & - \frac{383}{921} + \frac{143}{1000}i & - \frac{69}{374} - \frac{251}{741}i & - \frac{17}{888} + \frac{25}{266}i & \frac{1}{38} - \frac{182}{989}i & \frac{55}{112} - \frac{229}{986}i \\
- \frac{1}{165} - \frac{106}{393}i & - \frac{377}{869} + \frac{5}{876}i & - \frac{157}{475} + \frac{55}{224}i & \frac{12}{343} + \frac{87}{992}i & \frac{230}{519} - \frac{201}{692}i & \frac{12}{343} + \frac{87}{992}i & \frac{1}{786} + \frac{73}{254}i & \frac{377}{869} - \frac{5}{876}i \\
\frac{147}{725} - \frac{5}{16}i & \frac{349}{771} + \frac{102}{919}i & - \frac{23}{459} + \frac{68}{369}i & - \frac{175}{729} - \frac{184}{863}i & - \frac{23}{459} + \frac{68}{369}i & \frac{307}{644} + \frac{141}{508}i & - \frac{147}{725} + \frac{5}{16}i & \frac{41}{798} - \frac{134}{897}i \\
- \frac{313}{785} + \frac{259}{536}i & \frac{239}{836} + \frac{19}{921}i & - \frac{81}{377} + \frac{65}{992}i & \frac{25}{203} + \frac{71}{608}i & \frac{282}{773} - \frac{223}{997}i & \frac{25}{203} + \frac{71}{608}i & - \frac{303}{814} + \frac{109}{939}i & - \frac{239}{836} - \frac{19}{921}i \\
- \frac{278}{993} + \frac{12}{83}i & - \frac{1}{152} - \frac{95}{199}i & \frac{1}{65} + \frac{159}{659}i & \frac{191}{681} - \frac{5}{12}i & \frac{1}{65} + \frac{159}{659}i & \frac{22}{749} + \frac{127}{320}i & \frac{278}{993} - \frac{12}{83}i & \frac{173}{966} - \frac{103}{874}i
\end{pmatrix}.\]

By implementing step \ref{item: m calc}, we are able to calculate

\begin{align*}
    m_0=&\begin{pmatrix}
0 & - \frac{1}{2}i & 0 & 0 & 0 & 0 & 0 & \frac{1}{2}i \\
- \frac{1}{2}i & 0 & 0 & 0 & 0 & 0 & \frac{1}{2}i & 0 \\
0 & 0 & 0 & \frac{1}{2}i & 0 & \frac{1}{2}i & 0 & 0 \\
0 & 0 & \frac{1}{2}i & 0 & \frac{1}{2}i & 0 & 0 & 0 \\
0 & 0 & 0 & \frac{1}{2}i & 0 & \frac{1}{2}i & 0 & 0 \\
0 & 0 & \frac{1}{2}i & 0 & \frac{1}{2}i & 0 & 0 & 0 \\
0 & \frac{1}{2}i & 0 & 0 & 0 & 0 & 0 & - \frac{1}{2}i \\
\frac{1}{2}i & 0 & 0 & 0 & 0 & 0 & - \frac{1}{2}i & 0 \\
\end{pmatrix}
=\frac{1}{2}i\sigma_x\sigma_x\sigma_x- \frac{1}{2}i\sigma_z\sigma_z\sigma_x,
\end{align*}
and step \ref{item: log m} allows us to then calculate
\[\hspace*{-1cm} K_{0,0}=\begin{pmatrix}
- \frac{349}{992} + \frac{5}{362}i & 0 & \frac{63}{208} + \frac{58}{647}i & 0 & - \frac{165}{499} + \frac{73}{234}i & 0 & - \frac{254}{609} + \frac{412}{655}i & 0 \\
0 & \frac{103}{812} + \frac{87}{440}i & 0 & \frac{407}{974} + \frac{11}{628}i & 0 & \frac{117}{652} - \frac{482}{903}i & 0 & \frac{554}{901} - \frac{249}{911}i \\
\frac{15}{208} + \frac{165}{857}i & 0 & - \frac{35}{307} + \frac{441}{515}i & 0 & \frac{71}{624} + \frac{223}{533}i & 0 & \frac{109}{972} - \frac{98}{939}i & 0 \\
0 & \frac{95}{616} - \frac{83}{296}i & 0 & - \frac{533}{887} + \frac{68}{279}i & 0 & - \frac{185}{906} + \frac{149}{765}i & 0 & \frac{578}{977} - \frac{212}{973}i \\
- \frac{9}{980} - \frac{490}{967}i & 0 & - \frac{13}{46} + \frac{12}{53}i & 0 & \frac{55}{112} - \frac{291}{940}i & 0 & \frac{4}{911} + \frac{376}{717}i & 0 \\
0 & \frac{619}{993} + \frac{53}{239}i & 0 & - \frac{63}{452} - \frac{181}{974}i & 0 & \frac{526}{911} + \frac{68}{223}i & 0 & - \frac{92}{771} - \frac{237}{911}i \\
- \frac{394}{961} + \frac{461}{721}i & 0 & - \frac{50}{331} - \frac{1}{570}i & 0 & \frac{3}{7} - \frac{254}{873}i & 0 & - \frac{209}{579} - \frac{16}{399}i & 0 \\
0 & - \frac{16}{187} - \frac{561}{890}i & 0 & \frac{47}{114} - \frac{275}{647}i & 0 & \frac{142}{881} + \frac{155}{399}i & 0 & \frac{8}{31} + \frac{35}{997}i \\
\end{pmatrix}.\]
We then apply step~\ref{item: opt 2} to calculate 
\[K_{0,1}=\begin{pmatrix}
1 & 0 & 0 & 0 & 0 & 0 & 0 & 0 \\
0 & 1 & 0 & 0 & 0 & 0 & 0 & 0 \\
0 & 0 & 1 & 0 & 0 & 0 & 0 & 0 \\
0 & 0 & 0 & 1 & 0 & 0 & 0 & 0 \\
0 & 0 & 0 & 0 & 1 & 0 & 0 & 0 \\
0 & 0 & 0 & 0 & 0 & 1 & 0 & 0 \\
0 & 0 & 0 & 0 & 0 & 0 & 1 & 0 \\
0 & 0 & 0 & 0 & 0 & 0 & 0 & 1 \\
\end{pmatrix}
\qandq 
h^{(0)}=m_0.\]

We can now decompose the $K_{0,0}$ matrix. Step~\ref{item: m_2 calc} allows us to calculate
\begin{align*}
\hspace*{-1cm} m_1&=\begin{pmatrix}
- \frac{789}{959}i & 0 & \frac{64}{309} - \frac{75}{659}i & 0 & \frac{87}{788} + \frac{191}{981}i & 0 & \frac{456}{943} + \frac{33}{59}i & 0 \\
0 & \frac{789}{959}i & 0 & - \frac{64}{309} + \frac{75}{659}i & 0 & - \frac{87}{788} - \frac{191}{981}i & 0 & - \frac{456}{943} - \frac{33}{59}i \\
- \frac{64}{309} - \frac{75}{659}i & 0 & - \frac{91}{254}i & 0 & \frac{40}{577} + \frac{207}{653}i & 0 & - \frac{87}{788} + \frac{359}{864}i & 0 \\
0 & \frac{64}{309} + \frac{75}{659}i & 0 & \frac{91}{254}i & 0 & - \frac{40}{577} - \frac{207}{653}i & 0 & \frac{87}{788} - \frac{359}{864}i \\
- \frac{87}{788} + \frac{191}{981}i & 0 & - \frac{40}{577} + \frac{207}{653}i & 0 & - \frac{9}{16}i & 0 & - \frac{64}{309} + \frac{364}{829}i & 0 \\
0 & \frac{87}{788} - \frac{191}{981}i & 0 & \frac{40}{577} - \frac{207}{653}i & 0 & \frac{9}{16}i & 0 & \frac{64}{309} - \frac{364}{829}i \\
- \frac{456}{943} + \frac{33}{59}i & 0 & \frac{87}{788} + \frac{359}{864}i & 0 & \frac{64}{309} + \frac{364}{829}i & 0 & \frac{251}{417}i & 0 \\
0 & \frac{456}{943} - \frac{33}{59}i & 0 & - \frac{87}{788} - \frac{359}{864}i & 0 & - \frac{64}{309} - \frac{364}{829}i & 0 & - \frac{251}{417}i \\
\end{pmatrix}\\
\hspace*{-1cm} &=- \frac{129}{452}iII\sigma_z+\frac{74}{455}iI\sigma_x\sigma_z- \frac{327}{803}iI\sigma_z\sigma_z+\frac{299}{980}i\sigma_xI\sigma_z+\frac{333}{760}i\sigma_x\sigma_x\sigma_z+\frac{64}{309}i\sigma_x\sigma_y\sigma_z- \frac{87}{788}i\sigma_x\sigma_z\sigma_z+\frac{196}{709}i\sigma_y\sigma_x\sigma_z\\
\hspace*{-1cm} &- \frac{67}{553}i\sigma_y\sigma_y\sigma_z+\frac{87}{788}i\sigma_y\sigma_z\sigma_z- \frac{299}{980}i\sigma_zI\sigma_z- \frac{196}{709}i\sigma_z\sigma_x\sigma_z+\frac{64}{309}i\sigma_z\sigma_y\sigma_z+\frac{7}{40}i\sigma_z\sigma_z\sigma_z,
\end{align*}
and step~\ref{item: K_1 calc} shows us that
\[\hspace*{-1cm} K_1=\begin{pmatrix}
- \frac{13}{313} + \frac{296}{869}i & 0 & \frac{487}{823} + \frac{120}{619}i & 0 & - \frac{13}{765} - \frac{67}{930}i & 0 & \frac{195}{523} + \frac{113}{191}i & 0 \\
0 & - \frac{13}{313} + \frac{296}{869}i & 0 & \frac{487}{823} + \frac{120}{619}i & 0 & - \frac{13}{765} - \frac{67}{930}i & 0 & \frac{195}{523} + \frac{113}{191}i \\
- \frac{13}{765} + \frac{67}{930}i & 0 & - \frac{195}{523} + \frac{113}{191}i & 0 & \frac{13}{313} + \frac{296}{869}i & 0 & \frac{487}{823} - \frac{120}{619}i & 0 \\
0 & - \frac{13}{765} + \frac{67}{930}i & 0 & - \frac{195}{523} + \frac{113}{191}i & 0 & \frac{13}{313} + \frac{296}{869}i & 0 & \frac{487}{823} - \frac{120}{619}i \\
\frac{169}{303} - \frac{120}{619}i & 0 & - \frac{71}{467} - \frac{4}{305}i & 0 & \frac{481}{662} - \frac{11}{829}i & 0 & \frac{13}{765} + \frac{292}{925}i & 0 \\
0 & \frac{169}{303} - \frac{120}{619}i & 0 & - \frac{71}{467} - \frac{4}{305}i & 0 & \frac{481}{662} - \frac{11}{829}i & 0 & \frac{13}{765} + \frac{292}{925}i \\
- \frac{481}{662} - \frac{11}{829}i & 0 & \frac{13}{765} - \frac{292}{925}i & 0 & \frac{169}{303} + \frac{120}{619}i & 0 & \frac{71}{467} - \frac{4}{305}i & 0 \\
0 & - \frac{481}{662} - \frac{11}{829}i & 0 & \frac{13}{765} - \frac{292}{925}i & 0 & \frac{169}{303} + \frac{120}{619}i & 0 & \frac{71}{467} - \frac{4}{305}i \\
\end{pmatrix}.\]
We then remove the $I^{\otimes(n-1)}\otimes Z$ as describe in step~\ref{item: z phase}, obtaining
\[\hspace*{-1cm} \widehat{m}_1= \begin{pmatrix}
- \frac{403}{750}i & 0 & \frac{64}{309} - \frac{75}{659}i & 0 & \frac{87}{788} + \frac{191}{981}i & 0 & \frac{456}{943} + \frac{33}{59}i & 0 \\
0 & \frac{403}{750}i & 0 & - \frac{64}{309} + \frac{75}{659}i & 0 & - \frac{87}{788} - \frac{191}{981}i & 0 & - \frac{456}{943} - \frac{33}{59}i \\
- \frac{64}{309} - \frac{75}{659}i & 0 & - \frac{65}{892}i & 0 & \frac{40}{577} + \frac{207}{653}i & 0 & - \frac{87}{788} + \frac{359}{864}i & 0 \\
0 & \frac{64}{309} + \frac{75}{659}i & 0 & \frac{65}{892}i & 0 & - \frac{40}{577} - \frac{207}{653}i & 0 & \frac{87}{788} - \frac{359}{864}i \\
- \frac{87}{788} + \frac{191}{981}i & 0 & - \frac{40}{577} + \frac{207}{653}i & 0 & - \frac{23}{83}i & 0 & - \frac{64}{309} + \frac{364}{829}i & 0 \\
0 & \frac{87}{788} - \frac{191}{981}i & 0 & \frac{40}{577} - \frac{207}{653}i & 0 & \frac{23}{83}i & 0 & \frac{64}{309} - \frac{364}{829}i \\
- \frac{456}{943} + \frac{33}{59}i & 0 & \frac{87}{788} + \frac{359}{864}i & 0 & \frac{64}{309} + \frac{364}{829}i & 0 & \frac{874}{985}i & 0 \\
0 & \frac{456}{943} - \frac{33}{59}i & 0 & - \frac{87}{788} - \frac{359}{864}i & 0 & - \frac{64}{309} - \frac{364}{829}i & 0 & - \frac{874}{985}i \\
\end{pmatrix}\]
and
\[\widetilde{m}_1=\begin{pmatrix}
- \frac{129}{452}i & 0 & 0 & 0 & 0 & 0 & 0 & 0 \\
0 & \frac{129}{452}i & 0 & 0 & 0 & 0 & 0 & 0 \\
0 & 0 & - \frac{129}{452}i & 0 & 0 & 0 & 0 & 0 \\
0 & 0 & 0 & \frac{129}{452}i & 0 & 0 & 0 & 0 \\
0 & 0 & 0 & 0 & - \frac{129}{452}i & 0 & 0 & 0 \\
0 & 0 & 0 & 0 & 0 & \frac{129}{452}i & 0 & 0 \\
0 & 0 & 0 & 0 & 0 & 0 & - \frac{129}{452}i & 0 \\
0 & 0 & 0 & 0 & 0 & 0 & 0 & \frac{129}{452}i \\
\end{pmatrix}.\]

We now use step~\ref{item: seccond decomp} to obtain
\[\hspace*{-1cm} K_{11}=\begin{pmatrix}
\frac{103}{887} + \frac{205}{934}i & 0 & \frac{29}{38} + \frac{142}{641}i & 0 & - \frac{54}{167} - \frac{197}{891}i & 0 & \frac{257}{794} - \frac{191}{866}i & 0 \\
0 & \frac{103}{887} + \frac{205}{934}i & 0 & \frac{29}{38} + \frac{142}{641}i & 0 & - \frac{54}{167} - \frac{197}{891}i & 0 & \frac{257}{794} - \frac{191}{866}i \\
- \frac{509}{666} - \frac{81}{371}i & 0 & \frac{65}{562} + \frac{119}{540}i & 0 & - \frac{294}{907} + \frac{187}{853}i & 0 & - \frac{256}{789} - \frac{165}{754}i & 0 \\
0 & - \frac{509}{666} - \frac{81}{371}i & 0 & \frac{65}{562} + \frac{119}{540}i & 0 & - \frac{294}{907} + \frac{187}{853}i & 0 & - \frac{256}{789} - \frac{165}{754}i \\
- \frac{317}{977} - \frac{138}{629}i & 0 & \frac{294}{907} - \frac{75}{341}i & 0 & \frac{589}{771} - \frac{142}{649}i & 0 & \frac{97}{841} - \frac{179}{814}i & 0 \\
0 & - \frac{317}{977} - \frac{138}{629}i & 0 & \frac{294}{907} - \frac{75}{341}i & 0 & \frac{589}{771} - \frac{142}{649}i & 0 & \frac{97}{841} - \frac{179}{814}i \\
- \frac{257}{794} + \frac{11}{50}i & 0 & - \frac{281}{869} - \frac{80}{363}i & 0 & - \frac{114}{979} + \frac{205}{932}i & 0 & \frac{439}{575} - \frac{212}{959}i & 0 \\
0 & - \frac{257}{794} + \frac{11}{50}i & 0 & - \frac{281}{869} - \frac{80}{363}i & 0 & - \frac{114}{979} + \frac{205}{932}i & 0 & \frac{439}{575} - \frac{212}{959}i \\
\end{pmatrix}\]
and
\begin{align*}
f^{(0)}&=\begin{pmatrix}
\frac{641}{897}i & 0 & 0 & 0 & 0 & 0 & - \frac{371}{452}i & 0 \\
0 & - \frac{641}{897}i & 0 & 0 & 0 & 0 & 0 & \frac{371}{452}i \\
0 & 0 & - \frac{323}{452}i & 0 & \frac{1}{4}i & 0 & 0 & 0 \\
0 & 0 & 0 & \frac{323}{452}i & 0 & - \frac{1}{4}i & 0 & 0 \\
0 & 0 & \frac{1}{4}i & 0 & - \frac{323}{452}i & 0 & 0 & 0 \\
0 & 0 & 0 & - \frac{1}{4}i & 0 & \frac{323}{452}i & 0 & 0 \\
- \frac{371}{452}i & 0 & 0 & 0 & 0 & 0 & \frac{651}{911}i & 0 \\
0 & \frac{371}{452}i & 0 & 0 & 0 & 0 & 0 & - \frac{651}{911}i \\
\end{pmatrix}\\
&=- \frac{129}{452}i\sigma_x\sigma_x\sigma_z+\frac{121}{226}i\sigma_y\sigma_y\sigma_z+\frac{323}{452}i\sigma_z\sigma_z\sigma_z.
\end{align*}
Finally, with step~\ref{item: su(n-1)},
\[K^{(0)}=\begin{pmatrix}
- \frac{593}{790} - \frac{271}{823}i & - \frac{78}{839} + \frac{97}{878}i & - \frac{343}{986} - \frac{55}{761}i & \frac{162}{523} + \frac{236}{809}i \\
\frac{305}{988} - \frac{253}{864}i & - \frac{157}{451} + \frac{38}{535}i & \frac{91}{983} + \frac{41}{370}i & \frac{341}{455} - \frac{81}{244}i \\
- \frac{33}{355} - \frac{71}{643}i & \frac{322}{429} - \frac{302}{917}i & \frac{131}{423} - \frac{68}{233}i & \frac{295}{848} - \frac{6}{83}i \\
- \frac{345}{991} - \frac{66}{929}i & - \frac{296}{959} - \frac{237}{809}i & \frac{625}{834} + \frac{253}{762}i & - \frac{56}{605} + \frac{34}{307}i \\
\end{pmatrix},\]
\[K^{(1)}=\begin{pmatrix}
\frac{103}{887} - \frac{205}{934}i & - \frac{509}{666} + \frac{81}{371}i & - \frac{317}{977} + \frac{138}{629}i & - \frac{257}{794} - \frac{11}{50}i \\
\frac{29}{38} - \frac{142}{641}i & \frac{65}{562} - \frac{119}{540}i & \frac{294}{907} + \frac{75}{341}i & - \frac{281}{869} + \frac{80}{363}i \\
- \frac{54}{167} + \frac{197}{891}i & - \frac{294}{907} - \frac{187}{853}i & \frac{589}{771} + \frac{142}{649}i & - \frac{114}{979} - \frac{205}{932}i \\
\frac{257}{794} + \frac{191}{866}i & - \frac{256}{789} + \frac{165}{754}i & \frac{97}{841} + \frac{179}{814}i & \frac{439}{575} + \frac{212}{959}i \\
\end{pmatrix},\]
\[\widetilde{K}^{(0)}=\begin{pmatrix}
\frac{427}{445} - \frac{212}{753}i & 0 \\
0 & \frac{427}{445} + \frac{212}{753}i \\
\end{pmatrix}
\qandq
\varphi=\frac{408}{577} - \frac{408}{577}i.
\]

This gives us the decomposition
\[G\approx e^{i\varphi}K^{(0)}\otimes I \cdot e^{f^{(0)}}\cdot K^{(1)}\otimes I\cdot I^{\otimes 2}\otimes \widetilde{K}^{(0)}\cdot e^{h^{(0)}},\]
which we can visualize in Fig.~\ref{fig:exanple decomp}. Efficient algorithms then exist for implementing all resulting multi-qubit factors: $K^{(0)}$ and $K^{(1)}$ can be implemented following \cite{PhysRevA.69.032315, QiskitTwoQubitBasisDecomposer}, a processes we exemplify in Fig.~\ref{fig:k0_subdecomp}, while $e^{h^{(0)}}$ and $e^{f^{(0)}}$ can be implemented using methods from \cite{Mansky_2023}, as shown in Fig.~\ref{fig:h0_subdecomp}.
\begin{figure}[h]
    \centering
    \includegraphics[width=\linewidth]{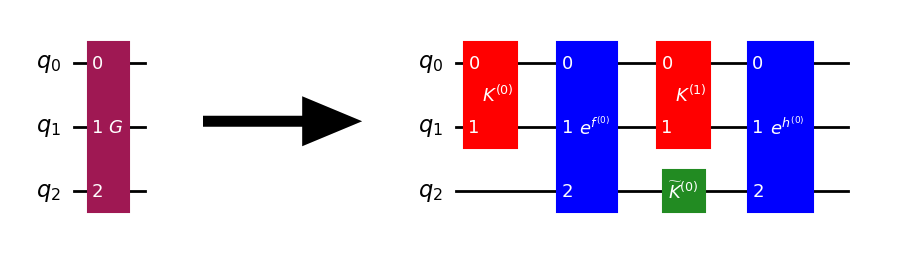}
    \caption{Decomposition of matrix G using our algorithm.}
    \label{fig:exanple decomp}
\end{figure}
\begin{figure}[h]
    \centering
    \includegraphics[width=\linewidth]{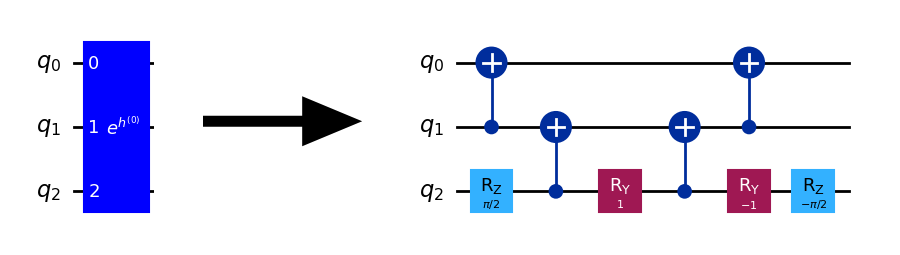}
    \caption{Decomposition of one of the abelian factors into single-qubit rotations and CNOTs using the method by \cite{Mansky_2023}.}
    \label{fig:h0_subdecomp}
\end{figure}
\begin{figure}[h]
    \centering
    \includegraphics[width=\linewidth]{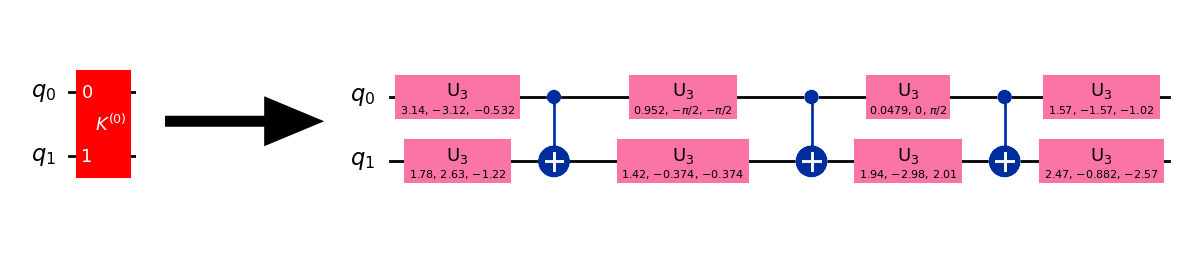}
    \caption{Decomposition of a two-qubit gate into single-qubit rotations and CNOTs using the Qiskit's Python library \cite{QiskitTwoQubitBasisDecomposer}.}
    \label{fig:k0_subdecomp}
\end{figure}

The decomposition errors are summarized in Table~\ref{tab: example}, once again showing the approximation accuracy and subspace fidelity achieved with our method.

\begin{table}[!h]
\centering
\begin{tabular}{l|llll}
         & $E_{a}$ &  $E_{s}(f^{(0)})$ & $E_{s}(h^{(0)})$ \\ \hline
G        & $3.4\cdot10^{-15}$                            & $2.4\cdot10^{-15}$                                 & $5.6\cdot10^{-6}     $
\end{tabular}
\caption{Errors of the decomposition of G.}
\label{tab: example}
\end{table}

\section{Phase correction}
\label{appendix: phase}

An involutive automorphism $\theta$  of  $\su(2^n)$ such that $\theta(g):=UgU^*$, for some matrix $U$, integrates to an automorphism $\Theta$ of the Lie group $\SU(2^n)$, with differential $\theta$ and which satisfies  \eqref{eq: Theta inv}, indeed $\Theta(G)=UGU^*$. Observe that there are elements of $\SU(2^n)$ which lie in the  image of $\Span\{iI\}$ under the exponential map; such elements are of the form $e^{i\varphi}I$, where $e^{i\varphi}$ is a $2^n$-th root of unity. Since 
\begin{equation*}
    \Theta(e^{i\varphi}I)
    =e^{i\varphi}I,
\end{equation*}
    we know from Proposition \ref{prop: m involution trick} that
\begin{equation*}
    \exp(2m)=\Theta(e^{-i\varphi}I)e^{i\varphi}I=I.
\end{equation*}
Therefore $m=0$. 

However, for the Khaneja-Glaser basis and the automorphism $\theta_X$, the elements $e^{i\varphi}I$ are not necessarily spanned by elements of the subalgebra $\mathfrak{k}_n$. As it happens, for those matrices we also have $\Theta_X(e^{i\varphi}I)=e^{i\varphi}I$, so the global phases are invariant in the decomposition and are `stored' in the matrix $K_{j,0}$, in the sense that
\begin{equation*}
    K_{j,0} =e^{i\varphi}\exp(k_{j,0}),
\end{equation*}
    with $k_{j,0}\in \su(2^{n-1})\otimes I$. Therefore, after removing the even rows and columns of $K_{j,0}$, we must extract the global phase to obtain the desired element in $\SU(2^{n-1})$.

\section{BCH expansion}
\label{app: BCH}

Here we will expand on how \cite{SaEarp2005} implemented steps \ref{item: comp. basis}-\ref{item: log m} of the decomposition, mainly finding $k\in\mathfrak{k}_n$ and $m\in\mathfrak{m}_n$ such that
\begin{equation}\label{eq: BCH basis}
    G=e^g=e^{k}e^m.
\end{equation}

Their approach is based on the Baker-Campbell-Hausdorff (BCH) expansion, which enables the approximation of $\log\left(e^ae^b\right)$. It can be computationally implemented using Dynkin's formula \cite{dynkin1949representation,10.1063/1.528242}
\begin{equation*}
    \log\left(e^ae^b\right) =\sum_{n=1}^{\infty} \frac{(-1)^{n-1}}{n} \sum_{\substack{r_1+s_1>0 \\ \vdots \\ r_n+s_n>0}} \frac{\left[a^{r_1} b^{s_1} a^{r_2} b^{s_2} \cdots a^{r_n} b^{s_n}\right]}{\left(\sum_{j=1}^n\left(r_j+s_j\right)\right) \cdot \prod_{i=1}^n r_{i} ! s_{i} !},
\end{equation*}
where the sum is performed over all nonnegative values of $s_i$ and $r_i$, and 
\begin{equation*}
    \left[a^{r_1} b^{s_1} \cdots a^{r_n} b^{s_n}\right]=[\underbrace{a,[a, \cdots[a}_{r_1},[\underbrace{b,[b, \cdots[b}_{s_1}, \cdots[\underbrace{a,[a, \cdots[a}_{r_n},[\underbrace{b,[b, \cdots b}_{s_n}]].
\end{equation*}
Applying the formula to  \eqref{eq: BCH basis}, one gets 
\begin{align*}
\begin{split}
    g & =k+m+\underbrace{\frac{1}{2}[k, m]}_{\in \mathfrak{m}_n}+\underbrace{\frac{1}{12}[k,[k, m]]}_{\in \mathfrak{m}_n}+\underbrace{\frac{1}{12}[m,[m, k]]}_{\in \mathfrak{k}_n}+\underbrace{\frac{1}{24}[k,[m,[k, m]]]}_{\in \mathfrak{k}_n}+\ldots, \\
& =P_{\mathfrak{k}}(k,m)+P_{\mathfrak{m}}(k,m),
\end{split}
\end{align*}
where
\begin{align*}
\begin{split}
P_{\mathfrak{k}}(k,m) & =k+\frac{1}{12}[m,[m, k]]+\frac{1}{24}[k,[m,[k, m]]]+\ldots \in \mathfrak{k}_n, \\
P_{\mathfrak{m}}(k,m) & =m+\frac{1}{2}[k, m]+\frac{1}{12}[k,[k, m]]+\ldots \in \mathfrak{m}_n .
\end{split}
\end{align*}
If one now defines
$\widetilde{P}(k,m)=P_{\mathfrak{k}}(k,m)-\mathrm{proj}_{\mathfrak{m}_n}(g)$, a root of $\widetilde{P}$ would also solve  \eqref{eq: BCH basis}.
Moreover, given $m$, we can find express $k=k(m)$ via
\begin{equation*}
    k=\log(Ge^{-m}),
\end{equation*}
using again the BCH expansion:
\begin{equation*}
    k=Q(m)=\log \left(e^g e^{-m}\right)=g-m-\frac{1}{2}[g, m]-\frac{1}{12}[g,[g, m]]+\frac{1}{12}[m,[m, g]]+\dots.
\end{equation*}
Furthermore, truncating
\begin{equation}
\label{eq: BCH polynomial}
    \widetilde{P}(Q(m),m)=P_{\mathfrak{k}}(Q(m),m)-\mathrm{proj}_{\mathfrak{m}_n}(g)
\end{equation}
sufficiently far down the series, it is not difficult to numerically find a root.
\end{document}